\documentclass[submission,copyright,creativecommons]{eptcs}

\usepackage{amsmath,amssymb,amsthm}
\usepackage{proof}
\usepackage{xspace}
\usepackage{wasysym} 
\usepackage{stmaryrd} 
\usepackage{color}
\usepackage{listingsutf8}
\usepackage[all]{xy}
\usepackage{subcaption}
\usepackage[utf8]{inputenc}
\usepackage[T1]{fontenc}

\newtheorem{theorem}{Theorem}

\newtheorem{lemma}[theorem]{Lemma}

\newcommand{\mixedProcesses}{{\mathcal{M}_0}}
\newcommand{\classicalProcesses}{\mathcal{C}}

\newcommand{\Keyword}[1]{\textsf{\upshape\small #1}\xspace}

\newcommand{\truek}{\Keyword{true}}
\newcommand{\falsek}{\Keyword{false}}

\newcommand{\ifk}{\Keyword{if}}
\newcommand{\thenk}{\Keyword{then}}
\newcommand{\elsek}{\Keyword{else}}

\newcommand{\PAR}{\mid}

\newcommand{\INACT}{\mathbf 0}
\newcommand{\BRANCHP}[4]{#1^{#2}{#3}.{#4}}
\newcommand{\CHOOSE}[4]{{#1}{#2}\sum_{#3}{#4}}
\newcommand{\SENDn}[2]{{#1}!{#2}}
\newcommand{\SEND}[2]{\SENDn{#1}{#2}.}

\newcommand{\RECEIVE}[2]{{#1}?{#2}.}
\newcommand{\SELECTn}[2]{#1 \lhd #2}
\newcommand{\SELECT}[2]{\SELECTn{#1}{#2}.}
\newcommand{\IF}[2]{\ifk\:#1\:\thenk\:#2\:\elsek\:}
\newcommand{\NR}[1]{(\nu #1)}

\newcommand{\TBRANCH}[4]{{#1}\rhd\{{#2}\colon{#3}\}_{#4}}

\newcommand{\sendp}{\SEND xvP}

\newcommand{\selp}{\SELECT xlP}
\newcommand{\branchp}{\TBRANCH{x}{l_i}{P_i}{i\in I}}    
\newcommand{\ifp}{\IF vPQ}

\newcommand{\rmap}[1]{\llparenthesis{#1}\rrparenthesis}

\newcommand{\translation}{\rmap{\cdot}: \mixedProcesses\longrightarrow \classicalProcesses}

\newcommand{\fragr}{\sum_{j\in J} \BRANCHP{l_i}{!}{v_{ij}}{P_{ij}}}
\newcommand{\frags}{\sum_{k\in K} \BRANCHP{l_i}{?}{y_{ik}}{P'_{ik}}}

\newcommand{\osred}{\rightarrow}                
\newcommand{\msred}{\Rightarrow}                
\newcommand{\transred}{\Rightarrow}                 
\newcommand{\labred}[1]{\stackrel{#1}{\longrightarrow}}
\newcommand{\subs}[2]{[{#1}/{#2}]}

\newcommand{\NDchoice}{\operatorname{NDChoice}}

\newcommand{\rulename}[1]{[\text{\sc #1}]\xspace}

\newcommand{\typerulename}[1]{\rulename{T-{#1}}}
\newcommand{\ttrue}{\typerulename{True}}
\newcommand{\tfalse}{\typerulename{False}}
\newcommand{\tunit}{\typerulename{Unit}}
\newcommand{\tvar}{\typerulename{Var}}
\newcommand{\tout}{\typerulename{Out}}
\newcommand{\ttout}{\typerulename{TOut}}

\newcommand{\tin}{\typerulename{In}}
\newcommand{\ttin}{\typerulename{TIn}}

\newcommand{\tinact}{\typerulename{Inact}}
\newcommand{\tpar}{\typerulename{Par}}
\newcommand{\tif}{\typerulename{If}}
\newcommand{\tres}{\typerulename{Res}}

\newcommand{\tbranch}{\typerulename{Branch}}
\newcommand{\tselect}{\typerulename{Sel}}

\newcommand{\tsubt}{\typerulename{Subt}}
\newcommand{\tchoice}{\typerulename{Choice}}

\newcommand{\reductionrulename}[1]{\rulename{R-{#1}}}
\newcommand{\rlincom}{\reductionrulename{LinCom}}
\newcommand{\runcom}{\reductionrulename{UnCom}}

\newcommand{\rlinlin}{\reductionrulename{LinLin}}
\newcommand{\rlinun}{\reductionrulename{LinUn}}
\newcommand{\runlin}{\reductionrulename{UnLin}}
\newcommand{\runun}{\reductionrulename{UnUn}}
\newcommand{\rcase}{\reductionrulename{Case}}
\newcommand{\rres}{\reductionrulename{Res}}
\newcommand{\rift}{\reductionrulename{IfT}}   
\newcommand{\riff}{\reductionrulename{IfF}}
\newcommand{\rstruct}{\reductionrulename{Struct}}
\newcommand{\rpar}{\reductionrulename{Par}}

\newcommand{\subt}{<:}
\newcommand{\branch}{\&}
\newcommand{\sel}{\!\oplus\!}
\newcommand{\un}{\Keyword{un}}
\newcommand{\lin}{\Keyword{lin}}
\newcommand{\Oplus}{\!\oplus\!} 

\newcommand{\typeconst}[1]{\Keyword{#1}}

\newcommand{\bool}{\typeconst{bool}}
\newcommand{\unit}{\typeconst{unit}}
\newcommand{\End}{\typeconst{end}}

\newcommand{\CHOICE}[4]{{#1}{#2}\{{#4}\}_{#3}}

\newcommand{\BRANCH}[4]{\BRANCHP{#1}{#2}{#3}{#4}}
\newcommand{\REC}[2]{\mu{#1}.{#2}}
\newcommand{\OUTn}[1]{!{#1}}
\newcommand{\INn}[1]{?{#1}}
\newcommand{\OUT}[1]{\OUTn{#1}.}
\newcommand{\IN}[1]{\INn{#1}.}

\newcommand{\rcdt}{\{l_i\colon T_i\}_{i\in I}}
\newcommand{\brancht}[1][\alpha]{\branch\rcdt}
\newcommand{\selectt}{\sel\rcdt}

\newcommand{\rect}{\REC aT}

\newcommand{\LIN}{\lin\,} 
\newcommand{\UN}{\un\,} 
\newcommand{\Empty}{\cdot} 

\newcommand{\csplit}{\circ}
\newcommand{\cupdate}{+}

\newcommand{\tequiv}{\approx}

\newcommand{\areEquiv}[3][]{{#2} \tequiv {#3}}
\newcommand{\areDualP}[2]{{#1}\:\bot\:{#2}} 
\newcommand{\isSubt}[3][]{{#2} \subt {#3}}
\newcommand{\areDual}[3][]{{#2}\:\bot\:{#3}}
\newcommand{\isValue}[3][\Gamma]{{#1} \vdash {#2}\colon{#3}}
\newcommand{\isProc}[2][\Gamma]{{#1} \vdash {#2}}
\newcommand{\isBranch}[3][\Gamma]{{#1} \vdash {#2} \colon {#3}}

\newcommand{\barb}[2][P]{{#1}\downarrow_{{#2}}}
\newcommand{\wbarb}[2][P]{{#1}\Downarrow_{#2}}

\newcommand{\grmeq}{\; ::= \;}
\newcommand{\grmor}{\; \mid \;}

\newcommand\Small{\small}

\definecolor{darkviolet}{rgb}{0.5,0,0.4}
\definecolor{darkgreen}{rgb}{0,0.4,0.2}
\definecolor{darkblue}{rgb}{0.1,0.1,0.9}
\definecolor{darkgrey}{rgb}{0.5,0.5,0.5}
\definecolor{lightblue}{rgb}{0.4,0.4,1}

\lstdefinestyle{eclipse}{
  breaklines=true,
  basicstyle=\sffamily\Small,
  emphstyle=\color{red}\bfseries,
  keywordstyle=\color{darkviolet}\bfseries,
  commentstyle=\color{darkgreen},
  stringstyle=\color{darkblue},
  numberstyle=\color{darkgrey},
  emphstyle=\color{red},
  morecomment=[s][\color{lightblue}]{/**}{*/},
  escapeinside={@}{@},
  showstringspaces=false,
  tabsize=2
}
\lstdefinelanguage{mixedsessions}{
  style=eclipse,
  morekeywords=[1]{unit,bool,int,end,rec,lin,un,true,false,if,then,else,def,new,select,case,type,of},
  literate=
    {inact}{\bfseries 0}1
    {|>}{$\rhd$}1
    {ell}{$\ell$}1
    {->}{$\rightarrow$}1
    {++}{$\oplus$}1
    {_1}{{$_1$}}1
    {_2}{{$_2$}}1
    {_3}{{$_3$}}1
    {_4}{{$_4$}}1,
  sensitive=true,
  breaklines=true,
  morecomment=[l]{--},
  morecomment=[l]{//},
}

\providecommand{\event}{PLACES 2020} 
\title{Mixed Sessions: the Other Side of the Tape}
\author{
  Filipe Casal
  \qquad\qquad
  Andreia Mordido
  \qquad\qquad
  Vasco T.\ Vasconcelos
\institute{LASIGE, Faculdade de Ciências, Universidade de Lisboa, Portugal}
\email{
  \hspace{-2.5em}
  fmrcasal@fc.ul.pt
  \qquad \quad\;
  afmordido@fc.ul.pt
  \qquad \qquad\quad
  vmvasconcelos@fc.ul.pt
}}
\def\titlerunning{Mixed Sessions: the Other Side of the Tape}
\def\authorrunning{Filipe Casal, Andreia Mordido \& Vasco T.\ Vasconcelos}
\def\copyrightholders{F. Casal, A. Mordido \& V. T.\ Vasconcelos}

\begin{document}
\maketitle

\begin{abstract}
  Vasconcelos et al.~\cite{ESOP2020} introduced side A of the tape:
  there is an encoding of classical sessions into mixed sessions. Here
  we present side B: there is translation of (a subset of) mixed
  sessions into classical sessions. We prove that the translation is a
  minimal encoding, according to the criteria put forward by Kouzapas
  et al.~\cite{DBLP:journals/iandc/KouzapasPY19}.
\end{abstract}



\section{Classical Sessions, Mixed Sessions}
\label{sec:intro}

Mixed sessions were introduced by Vasconcelos et al.~\cite{ESOP2020}
as an extension of classical session
types~\cite{DBLP:journals/acta/GayH05,DBLP:conf/esop/HondaVK98,DBLP:journals/iandc/Vasconcelos12}.
They form an interesting point in the design space of session-typed
systems: an extremely concise process calculus (four constructors only)
that allows the natural expression of algorithms quite cumbersome to
write in classical sessions.
The original paper on mixed sessions~\cite{ESOP2020} shows that there
is an encoding of classical sessions into mixed sessions. This
abstract shows that the converse is also true for a fragment of mixed
sessions.

A translation of mixed sessions into classical sessions would allow to
leverage the tools available for the latter: one could program in
mixed sessions, translate the source code into classical sessions,
check the validity of the source code against the type system for the
target language, and run the original program under an interpreter for
classical sessions (SePi~\cite{DBLP:conf/sefm/FrancoV13}, for
example). A mixed-to-classical encoding would further allow a better
understanding of the relative expressiveness of the two languages.

%


Processes in classical binary
sessions~\cite{DBLP:journals/acta/GayH05,DBLP:conf/concur/Honda93,DBLP:conf/esop/HondaVK98,DBLP:conf/parle/TakeuchiHK94}
(here we follow the formulation
in~\cite{DBLP:journals/iandc/Vasconcelos12}) communicate by exchanging
messages on bidirectional channels. We introduce classical sessions by
means of a few examples. Each channel is denoted by its two ends and
introduced in a process \lstinline|P| as
\lstinline|(new xy)P|. Writing a value \lstinline|v| on channel end
\lstinline|x| and continuing as \lstinline|P| is written as
\lstinline|x!v.P|. Reading a value from a channel end \lstinline|y|,
binding it to variable \lstinline|z| and continuing as \lstinline|Q|
is written as \lstinline|y?z.Q|. When the two processes get together
under a new binder that ties together the two ends of the channel,
such as in
\begin{lstlisting}
    (new xy) x!v.P | y?z.Q
\end{lstlisting}
value \lstinline|v| is communicated from the \lstinline|x| channel end
to the \lstinline|y| end. The result is process
\lstinline@(new xy) P | Q[v/z]@, where notation \lstinline|Q[v/z]|
denotes the result of replacing \lstinline|v| for \lstinline|z| in
\lstinline|Q|.

Processes may also communicate by offering and selecting options in
choices. The different choices are denoted by labels, \lstinline|ell|
and \lstinline|m| for example. To select choice \lstinline|ell| on
channel end \lstinline|x| and continue as \lstinline|P| we write
\lstinline|x select ell.P|. To offer a collection of options at
channel end \lstinline|y| and continue with appropriate continuations
\lstinline|Q| and \lstinline|R|, we write
\lstinline|case y of {ell -> Q, m -> R}|. When \lstinline|select| and
\lstinline|case| processes are put together under a \lstinline|new|
that binds together the two ends of a channel, such as in
\begin{lstlisting}
    (new xy) x select ell.P | case y of {ell -> Q, m -> R}
\end{lstlisting}
branch \lstinline|Q| is selected in the \lstinline|case| process. The
result is the process \lstinline@(new xy) P | Q@. Selecting a choice
is called an internal choice, offering a collection of choices is
called an external choice.
We thus see that classical sessions comprise four atomic interaction
primitives. Furthermore, choices are directional in the sense that one
side offers a collection of possibilities, the other selects one of them.

To account for unbounded behavior classical sessions count with
replication: an input process that yields a new copy of itself after
reduction, written \lstinline|y*?z.Q|. A process of the form
\begin{lstlisting}
    (new xy) x!v.P | y*?z.Q
\end{lstlisting}
reduces to \lstinline@(new xy) P | Q[v/z] | y*?z.Q@. If we use the
\lstinline|lin| prefix to denote an ephemeral process and the
\lstinline|un| prefix to denote a persistent process, an alternative
syntax for the above process is \lstinline@(new xy) lin x!v.P | un y?z.Q@.


Mixed sessions blur the distinction between internal and external
choice.  Under a unified language construct---mixed choice---processes
may non-deterministically select one choice from a multiset of output
choices, or branch on one choice, again, from a multiset of possible
input choices.
Together with an output choice, a value is (atomically) sent; together
with an input choice, a value is (again, atomically) received,
following various proposals in the literature~\cite{DBLP:conf/concur/DemangeonH11,DBLP:journals/iandc/Sangiorgi98,DBLP:conf/ecoop/Vasconcelos94}.
The net effect is that the four common operations on session
types---output, input, selection, and branching---are effectively
collapsed into one: mixed choice.
Mixed choices can be labelled as ephemeral (linear, consumed by
reduction) or persistent (unrestricted, surviving reduction),
following conventional versus replicated inputs in some versions of
the pi-calculus~\cite{DBLP:journals/mscs/Milner92}.
Hence, in order to obtain a core calculus, all we have to add is name
restriction, parallel composition, and inaction (the terminated
process), all standard in the pi-calculus.

We introduce mixed sessions by means of a few examples.
Processes communicate by offering/selecting choices with the same label and opposite
polarities.

\begin{lstlisting}
    (new xy) lin x (m!3.P + n?z.Q) | lin y (m?w.R + n!5.S + p!7.T)
\end{lstlisting}
The above processes communicate over the channel with ends named
\lstinline|x| and \lstinline|y| and
reduce in one step along label \lstinline|m| to
\lstinline@(new xy) P | R[3/w]@ or along label \lstinline|n| to
\lstinline@(new xy) Q[5/z] | S@.
%

Non-determinism in mixed sessions can be further achieved by allowing
duplicated labels in choices. An example in which a 3 or a 5 is
non-deterministically sent over the channel is
\begin{lstlisting}
    (new xy) lin x (m!3.P + m!5.Q) | lin y (m?z.R)
\end{lstlisting}
This process reduces in one step to either \lstinline*(new xy) P | R[3/z]* or
\lstinline*(new xy) Q | R[5/z]*.
%
%
%
Unrestricted behavior in choices is achieved by the \lstinline|un|
qualifier in the choice syntax.
\begin{lstlisting}
    (new xy) un x (m!3.P + m!5.P) | un y (m?z.Q)
\end{lstlisting}
This process reduces to itself together with either of the choices
taken,
\par
\begin{minipage}{.4\textwidth}
  \centering
\begin{lstlisting}
(new xy)
un x (m!3.P + m!5.P) |
un y (m?z.Q) |
P | Q[3/z]
\end{lstlisting}
  \label{fig:prob1_6_2}
\end{minipage}%
\begin{minipage}{0.1\textwidth}
  or
\end{minipage}
\begin{minipage}{0.4\textwidth}
  \centering
\begin{lstlisting}
(new xy)
un x (m!3.P + m!5.P) |
un y (m?z.Q) |
P | Q[5/z]
\end{lstlisting}
\end{minipage}

The complete set of definitions for the syntax, operational semantics,
and type system for mixed sessions are in appendix,
Figures~\ref{fig:mixed-sessions} to~\ref{fig:mixed-sessions3}. For
technical details and main results, we direct the reader to
reference~\cite{ESOP2020}.
The complete set of definitions for the syntax, operational semantics,
and type system for classical sessions are in appendix,
Figure~\ref{fig:classical-sessions}. For further details, we refer the
reader to
references~\cite{DBLP:journals/iandc/Vasconcelos12,ESOP2020}.

\section{Mixed Sessions as Classical Sessions}
\label{sec:embedding}

This section shows that a subset of the language of mixed sessions can
be embedded in that of classical sessions.
We restrict our attention to choices that reduce against choices with
the same qualifier, that is, we do not consider the case where an
ephemeral ($\lin$) process reduces against a persistent ($\un$)
one. For this reason, we assume that a process and its type always
have the same $\lin/\un$ qualifier.

One of the novelties in mixed sessions is the possible presence of
duplicated label-polarity pairs in choices. This introduces a form of
non-determinism that can be easily captured in classical sessions.
The $\mathsf{NDchoice}$ classical session process creates a race
condition on a new channel with endpoints $s,t$ featuring multiple
selections on the $s$ endpoint, for only one branch on the $t$
endpoint. This guarantees that exactly one of the branches is
non-deterministically selected. The remaining selections must eventually
be garbage collected.
We assume that $\prod_{1\le i\le n} Q_i$ denotes the process
$Q_1\PAR\dots\PAR Q_n$ for $n>0$, and that $\Pi$ binds tighter than
the parallel composition operator.
\begin{equation*}
\NDchoice\{P_i\}_{i\in I} = \NR{st} \left(\prod _{i\in I}
    \SELECT{s}{l_i}{\INACT}  \PAR
    \TBRANCH{t}{l_i}{P_i}{i\in I}\right)
\end{equation*}

The type $S$ of channel end $s$ is of the form
$\CHOICE \un \oplus {i\in I} {l_i\colon S}$, an equation that can be
solved by type $\REC a {\CHOICE \un \oplus {i\in I} {l_i\colon a}}$,
and which SePi abbreviates to $*\oplus \{l_i\}_{i\in I}$. The
qualifier must be $\un$ because $s$ occurs in multiple threads in
$\NDchoice$; recursion arises because of the typing rules for
processes reading or writing in unrestricted channels.

Equipped with $\NDchoice$ we describe the translation of mixed
sessions to classical sessions via variants of the examples in
Section~\ref{sec:intro}. All examples fully type check and run in
SePi~\cite{DBLP:conf/sefm/FrancoV13}.
To handle duplicated label-polarity pairs in choices, we organize
choice processes by label-polarity fragments. Each such fragment
represents a part of a choice operation where all possible outcomes
have the same label and polarity. When a reduction occurs, one of the
branches is taken, non-deterministically, using the $\NDchoice$
operator.  After a non-deterministic choice of the branch, and
depending on the polarity of the fragment, the process continues by
either writing on or reading from the original channel.

The translation of choice processes is guided by their types. For each
choice we need to know its qualifier ($\lin, \un$) and its view
($\oplus, \&$), and this information is present in types alone.

\begin{figure}[t]
  \centering
  \begin{minipage}{.45\textwidth}
\begin{lstlisting}
new x y: lin&{m: !int.end,
              n: ?bool.end}

// lin x (m!3.0 + n?w.0)
case x of
    m -> new s_1 t_1: *+{ell}
         s_1 select ell |
         case t_1 of
             ell -> x!3
    n -> new s_2 t_2: *+{ell}
         s_2 select ell |
         case t_2 of
             ell -> x?w
 \end{lstlisting}
  \end{minipage}%
  \begin{minipage}{0.1\textwidth}
  $\PAR$
  \end{minipage}
  \begin{minipage}{0.4\textwidth}
      \centering
      \begin{lstlisting}


// lin y (m?z.0)
new s_3 t_3: *+{ell}
s_3 select ell |
case t_3 of
    ell -> y select m.
         new s_4 t_4: *+{ell}
         s_4 select ell |
         case t_4 of
             ell -> y?z
    \end{lstlisting}
  \end{minipage}
  \caption{Translation of \lstinline*(new xy)(lin x (m!3.0 + n?w.0) | lin y (m?z.0))*}
  \label{fig:example1_translation}
\end{figure}


%
Figure~\ref{fig:example1_translation} shows the translation of the
mixed process
\lstinline*(new xy)(lin x (m!3.0 + n?w.0) | lin y (m?z.0))*, where
\lstinline|x| is of type \lstinline|lin&{m!int.end, n?bool.end}|.
The corresponding type in classical sessions is
\lstinline|lin&{m:!int.end, n:?bool.end}|, which should not come as a
surprise.
Because channel end \lstinline|x| is of an external choice type
(\lstinline|&|), the choice on \lstinline|x| is encoded as a
\lstinline|case| process. The other end of the channel, \lstinline|y|,
is typed as an internal choice (\lstinline|++|) and is hence
translated as a \lstinline|select| process.
Occurrences of the $\NDchoice$ process appear in a degenerate form,
always applied to a single branch. We have four of them: three for
each of the branches in \lstinline|case| processes
(\lstinline|s_1t_1|, \lstinline|s_2t_2|, and \lstinline|s_4t_4)| and
one for the external choice in the mixed session process
(\lstinline|s_3t_3|).

In general, an external choice is translated into a classical
branching (\lstinline|case|) over the unique labels of the fragments
of the process, but where the polarity of each label is inverted.  The
internal choice, in turn, is translated as (possibly nondeterministic
collection of) classical \lstinline|select| process but keeps the
label polarity.
This preserves the behavior of the original process: in mixed choices,
a reduction occurs when a branch $\BRANCHP l!vP$ matches another
branch $\BRANCHP l?zQ$ with the same label but with dual polarity
($l^!$ against $l^?$), while in a classical session the labels alone
must match ($l$ against $l$). Needless to say, we could have followed
the strategy of dualizing internal choices rather than external.

If we label reduction steps with the names of the channel ends on
which they occur, we can see that, in this case a $\labred{xy}$
reduction step in mixed sessions is mimicked by a long series of
classical reductions, namely
$ \labred{s_3t_3} \labred{xy} \labred{s_1t_1} \labred{s_4t_4}$
$ \labred{xy} $ or
$ \labred{s_3t_3} \labred{xy} \labred{s_4t_4} \labred{s_1t_1}
\labred{xy} $.  
Notice the three reductions to resolve non-determinism
(on $s_it_i$) and the two reductions on $xy$ to encode branching
followed by message passing, an atomic operation in mixed sessions.


\begin{figure}
  \centering
  \begin{minipage}{.4\textwidth}
\begin{lstlisting}
new x y: lin&{m: !int.end}

// lin x (m!3.0 + m!5.0)
case x of
    m -> new s_1 t_1: *+{ell_1,ell_2}
        s_1 select ell_1 |
        s_1 select ell_2 |
        case t_1 of
            ell_1 -> x!3
            ell_2 -> x!5
 \end{lstlisting}
  \end{minipage}%
  \begin{minipage}{0.1\textwidth}
  $\quad\;\;\PAR$
  \end{minipage}
  \begin{minipage}{0.4\textwidth}
      \centering
      \begin{lstlisting}



// lin y (m?z.0)
new s_2 t_2: *+{ell}
s_2 select ell |
case t_2 of
    ell -> y select m.
        new s_3 t_3: *+{ell}
        s_3 select ell |
        case t_3 of
            ell -> y?z
    \end{lstlisting}
  \end{minipage}
  \caption{Translation of \lstinline*(new xy)(lin x (m!3.0 + m!5.0) | lin y (m?z.0))*}
  \label{fig:example2_translation}
\end{figure}


Figure~\ref{fig:example2_translation} shows an example of a mixed
choice process with a duplicated label-polarity pair,
\lstinline|m!|. If we assign to \lstinline|x| type
\lstinline|lin&{m!int}|, then we know that the choice on \lstinline|x|
is encoded as \lstinline|case| and that on \lstinline|y| as
\lstinline|select|.
In this case, the $\NDchoice$ operator is applied in a non-degenerate
manner to decide whether to send the values 3 or 5 on \lstinline|x|
channel end, by means of channel \lstinline|s_1t_1|.
Again we can see that the one step reduction on channel \lstinline|xy|
in the original mixed session process originates a sequence of five
reduction steps in classical sessions, namely
$
\labred{s_2t_2}
\labred{xy}
\labred{s_1t_1}
\labred{s_3t_3}
\labred{xy}
$
or
$
\labred{s_2t_2}
\labred{xy}
\labred{s_3t_3}
\labred{s_1t_1}
\labred{xy}
$.
In this case, however, the computation is non-deterministic: the last
reduction step may carry integer \lstinline|3| or \lstinline|5|.

\begin{figure}[t]
  \centering
  \begin{minipage}{.45\textwidth}
\begin{lstlisting}
type Unr = lin&{m: !integer.end}
new x y: *?Unr

// un x (m!3.0 + m!5.0)
new u_1 v_1: *!()
u_1!() |
v_1*?(). x?a.
case a of
    m -> new s_1 t_1: *+{ell_1,ell_2}
         s_1 select ell_1 |
         s_1 select ell_2 |
         case t_1 of
             ell_1 -> a!3 . u_1!()
             ell_2 -> a!5 . u_1!()
 \end{lstlisting}
  \end{minipage}%
  \begin{minipage}{0.1\textwidth}
  $\quad\;\;\PAR$
  \end{minipage}
  \begin{minipage}{0.4\textwidth}
      \centering
      \begin{lstlisting}


// un y (m?z.0)
new u_2 v_2: *!()
u_2!() |
v_2*?().
new s t: *+{ell}
s select ell |
case t of
    ell -> new a b: Unr
        y!a . b select m .
        new s_2 t_2: *+{ell}
        s_2 select ell |
        case t_2 of
            ell -> b?z . u_2!()
    \end{lstlisting}
  \end{minipage}
  \caption{Translation of \lstinline*(new xy)(un x (m!3.0 + m!5.0) | un y (m?z.0))*}
  \label{fig:example3_translation}
\end{figure}


Figure~\ref{fig:example3_translation} shows the encoding of mixed
choices on unrestricted channels. The mixed choice process is that of
Figure~\ref{fig:example2_translation} only that the two ephemeral
choices (\lstinline|lin|) have been replaced by their persistent
counterparts (\lstinline|un|).
The novelty, in this case, is the loops that have been created around the
\lstinline|case| and the \lstinline|select| process. Loops in
classical sessions can be implemented with a replicated input: a
process of the form \lstinline|v*?x.P| is a persistent process that,
when invoked with a value \lstinline|v| becomes the parallel
composition \lstinline@P[v/x] | v*?x.P@. The general form of the loops
we are interested in are \lstinline@(new uv : *!())(u!() | v*?x.P)@,
where \emph{continue calls} in process \lstinline|P| are of the form
\lstinline|u!()|. The contents of the messages that control the loop
are not of interest and so we use the unit type \lstinline|()|, so
that \lstinline|u| is of type \lstinline|*!()|.
We can easily see the calls \lstinline|u_1!()| and \lstinline|u_2!()|
in the last lines in Figure~\ref{fig:example3_translation}, reinstating
the unrestricted choice process.
In this case, one step reduction in mixed sessions corresponds to
a long sequence of transitions in their encodings.
%


\begin{figure}[t]
  \begin{align*}
    \rmap{\CHOICE \LIN\Oplus {i\in I} {\BRANCH {l_i}\star {S_i}{T_i}}} &=
      \CHOICE{\LIN}{\Oplus}{i\in I}{{{l_i}^\star \colon \lin{\star_i} \rmap{S_i}.\rmap{T_i}}}
    \\
      \rmap{\CHOICE \LIN\& {i\in I} {\BRANCH {l_i}\star {S_i}{T_i}}} &=
      \CHOICE{\LIN}{\&}{i\in I}{{{l_i}^\bullet \colon \lin{\star_i} \rmap{S_i}.\rmap{T_i}}}
      && \text{where } \areDualP{\star_i}\bullet_i
    \\
    \rmap{\CHOICE \UN\Oplus {i\in I} {\BRANCH {l_i}{\star} {S_i} {T_i}}} &=
      \REC b{\un\OUT{(\CHOICE \LIN\Oplus {i\in I} {{l_i}^\star \colon \lin{\star_i}\rmap{S_i}.\End})}} b
    && \text{where }
    T_i \tequiv \CHOICE \UN\Oplus {i\in I} {\BRANCH {l_i}{\star} {S_i} {T_i}}
    \\
    \rmap{\CHOICE \UN\& {i\in I} {\BRANCH {l_i}{\star} {S_i} {T_i}}} &=
    \REC b{\un\IN{(\CHOICE \LIN\& {i\in I} {{l_i}^\bullet \colon \lin{\star_i}\rmap{S_i}.\End})}} b
    && \text{where } \areDualP{\star_i}\bullet_i \text{ and }
    T_i \tequiv \CHOICE \UN\& {i\in I} {\BRANCH {l_i}{\star} {S_i} {T_i}}
\end{align*}
(Homomorphic for $\End$, $\unit$, $\bool$, $\rect$, and $a$)
  \caption{Translating mixed session types to traditional session types}
   \label{fig:embedding}
\end{figure}


We now present translations for types and processes in general.  The
translation of mixed choice session types into classical session types
is in Figure~\ref{fig:embedding}.
In general, the (atomic) branch-communicate nature of mixed session
types, $\{l_i^\star{S_i}\}$, is broken in its two
parts, $\{l_i\colon\star S_i\}$, branch first, communicate after.
In mixed sessions, choice types are labelled by label-polarity
pairs ($l^!$ or $l^?$); in classical session choices are labelled by
labels alone. Because we want the encoding of a label $l^!$ to match
the encoding of $l^?$, we must dualize one of them. We arbitrarily
chose do dualize the labels in the $\&$ type.
The typing rules for classical unrestricted processes of type
$S = \CHOICE \UN\sharp {i\in I} {\BRANCH {l_i}\star {S_i}{T_i}}$
require $T_i$ to be equivalent ($\tequiv$) to $S$ itself. We take
advantage of this restriction when translating $\un$ types.


\begin{figure}[t!]
  \begin{multline*}
    \rmap{\isProc{\lin x\sum_{i\in I} (\fragr + \frags)}} =
    x \rhd \{l_i^? \colon
            \NDchoice\{\SEND{x}{v_{ij}}{\rmap{\isProc[\Gamma_3,
                  x\colon T_i]{P_{ij}}}}\}_{j\in J},
    \\
     l_i^! \colon
      \NDchoice\{\RECEIVE{x}{y_{ik}}{\rmap{\isProc[\Gamma_2\csplit\Gamma_3,x\colon
            T'_i, y_{ik}\colon S'_i]{P'_{ik}}}}\}_{k\in K}
      \}_{i\in I}
  \end{multline*}
where $\Gamma = \Gamma_1 \csplit \Gamma_2  \csplit \Gamma_3$
and
$\isValue[\Gamma_1]{x}{\CHOICE{\lin}{\&}{i\in
          I}{\BRANCH{l_i}{!}{S_i}{T_i}, \BRANCH{l_i}{?}{S'_i}{T'_i}}}$
and
$\isValue[\Gamma_2]{v_{ij}}{S_i}$.
%
  \begin{multline*}
    \rmap{\isProc{\LIN x\sum_{i\in I} (\fragr + \frags)}} =
        \mathsf{NDChoice}\{
    \SELECT{x}{l_i^!}\NDchoice\{\SEND{x}{v_{ij}}{\rmap{\isProc[\Gamma_3,
           x\colon T_i]{P_{ij}}}}\}_{j\in J},
    \\
    \SELECT{x}{l_i^?}\NDchoice\{\RECEIVE{x}{y_{ik}}{\rmap{\isProc[\Gamma_2\csplit\Gamma_3,x\colon
            T'_i, y_{ik}\colon S'_i]{P'_{ik}}}}\}_{k\in K}
            \}_{i\in I}
  \end{multline*}
where $\Gamma = \Gamma_1 \csplit \Gamma_2 \csplit \Gamma_3$
and
$\isValue[\Gamma_1]{x}{\CHOICE{\lin}{\oplus}{i\in
          I}{\BRANCH{l_i}{!}{S_i}{T_i}, \BRANCH{l_i}{?}{S'_i}{T'_i}}}$
and
$\isValue[\Gamma_2]{v_{ij}}{S_i}$.
%
  \begin{multline*}
    \rmap{\isProc{\UN x\sum_{i\in I} (\fragr + \frags)}} = \NR{uv}
    ( \SENDn{u}{()} \PAR\UN\RECEIVE{v}{\_}\RECEIVE{x}{a}
    \\
    a \rhd \{l_i^? \colon
            \NDchoice\{\SEND{a}{v_{ij}}{( \SENDn{u}{()} \PAR {\rmap{\isProc[\Gamma]
            {P_{ij}}}}})\}_{j\in J}
    ,\\
    l_i^! \colon
      \NDchoice\{\RECEIVE{a}{y_{ik}}
      (\SENDn{u}{()} \PAR
      {\rmap{\isProc[\Gamma, y_{ik}\colon S'_i]{P'_{ik}}}})\}_{k\in K}
      \}_{i\in I} ) \qquad\qquad\quad\;\;\,
  \end{multline*}
where $\un(\Gamma)$
and
$\isValue x {\CHOICE{\un}{\&}{i\in I}{\BRANCH{l_i}{!}{S_i}{T_i}, \BRANCH{l_i}{?}{S'_i}{T_i'}}}$
and
$\isValue{v_{ij}}{S_i}$
and
$T_i \tequiv T'_i \tequiv \CHOICE \un \sharp {i\in I} {\BRANCH{l_i}{!}{S_i}{T_i},
  \BRANCH{l_i}{?}{S'_i}{T_i'}}$.
  \begin{multline*}
    \rmap{\isProc{\UN x\sum_{i\in I} (\fragr + \frags)}} = \NR{uv}
          ( \SENDn{u}{()} \PAR \UN\RECEIVE{v}{\_}\\
      \mathsf{NDChoice}\{
     \NR{ab} \SEND{x}{a} \SELECT{b}{l_i^!}\NDchoice\{\SEND{b}{v_{ij}}
     (\SENDn{u}{()} \PAR {\rmap{\isProc[\Gamma
               ]{P_{ij}}})}\}_{j\in J},
\\
            \NR{ab} \SEND{x}{a} \SELECT{b}{l_i^?}
            \NDchoice\{\RECEIVE{b}{y_{ik}}
            (\SENDn{u}{()} \PAR
            {\rmap{\isProc[\Gamma, y_{ik}\colon S_i']{P'_{ik}}})}\}_{k\in K}
            \}_{i\in I} )\quad\;\;\,
 \end{multline*}
where $\un(\Gamma)$
and
$\isValue x {\CHOICE{\un}{\oplus}{i\in I}{\BRANCH{l_i}{!}{S_i}{T_i}, \BRANCH{l_i}{?}{S'_i}{T'_i}}}$
and
$\isValue{v_{ij}}{S_i}$
and
$T_i \tequiv T'_i \tequiv \CHOICE \un \sharp {i\in I} {\BRANCH{l_i}{!}{S_i}{T_i},
  \BRANCH{l_i}{?}{S'_i}{T'_i}}$.
\begin{align*}
  \rmap{\isProc{\NR{xy}P}} =&\, \NR{xy}\rmap{\isProc[\Gamma,x\colon S,y\colon T]{P}}
  &&\text{where } \areDual[\Empty]{S}{T}
  \\
  \rmap{\isProc[\Gamma_1\csplit\Gamma_2]{P_1\PAR P_2}} =&\,
  \rmap{\isProc[\Gamma_1]{P_1}}
  \PAR
  \rmap{\isProc[\Gamma_2]{P_2}}
  \\
  \rmap{\isProc{\INACT}} =&\, \INACT
  \\
  \rmap{\isProc[\Gamma_1\csplit\Gamma_2]{\IF{v}{P_1}P_2}} =&\,
  \IF{v}{\rmap{\isProc[\Gamma_2]{P_1}}}\rmap{\isProc[\Gamma_2]{P_2}}
  &&\text{where } \isValue[\Gamma_1]{v}{\bool}
\end{align*}
  \caption{Translating mixed session processes to classical session processes}
  \label{fig-translation-mixed2classical}
\end{figure}


The translation of mixed choice processes is in
Figure~\ref{fig-translation-mixed2classical}. Since the translation is
guided by the type of the process to be translated, we also provide
the typing context to the translation function, hence the notation
$\rmap{\isProc P}$.
Because label-polarity pairs may be duplicated in choice processes, we
organize such processes in label-polarity fragments,
so that a process of the form
$\CHOOSE qx{i\in I}{\BRANCHP{l_i}{\star_i}{v_i}{P_i}}$ (where
$q::=\lin\mid\un$ and $\star::=!\mid?$) can be written as
$\CHOOSE qx{i\in I}{(\fragr + \frags)}$. Each label-polarity fragment
($l^!_i$ or $l^?_i$) groups together branches with the same label and
the same polarity. Such fragments may be empty for external choices,
for not all label-polarity pairs in an external choice type need to be
covered in the corresponding process (internal choice processes do not
need to cover all choices offered by the external counterpart).
The essence of the translation is discussed in the three examples
above.

We distinguish four cases for choices, according to qualifiers ($\lin$
or $\un$) and views ($\oplus$ or $\&$) in types. In all of them an
$\NDchoice$ process takes care of duplicated label-polarity pairs in
branches.
Internal choice processes feature an extra occurrence of $\NDchoice$
to non-deterministically select between output and input \emph{on the
  same label}.
Notice that external choice must still accept both choices, so that it
is not equipped with an $\NDchoice$.
Finally, unrestricted mixed choices require the encoding of a loop,
accomplished by creating a new channel for the effect ($uv$),
installing a replicated input $\UN\RECEIVE{v}{\_}P$ at one end of the
channel, and invoking the input once to ``start'' the loop and again
at the end of the interaction on channel end $x$. The calls are all
accomplished with processes of the form $\SENDn{u}{()}$. The contents
of the messages are of no interest and so we use the unit value $()$.

Following the encoding for types, the encoding for external choice
processes exchanges the polarities of choice labels: a label $l_i^!$
in mixed sessions is translated into $l_i^?$, and vice-versa, in the
cases for $\lin\&$ and $\un\&$ choices. This allows reduction to
happen in classical sessions, where we require an exact match between
the label of the \lstinline|select| process and that of the
\lstinline|case| process.


\section{A Minimal Encoding}
\label{sec:correspondences}

This section covers typing and operational correspondences; we follow
Kouzapas et al.~\cite{DBLP:journals/iandc/KouzapasPY19} criteria for
typed encodings, and aim at a minimal encoding.

Let $\mathcal{C}$ range over classical processes, and
$\mathcal{M}_0$
range over the fragment of mixed choice processes where $\lin$
processes only reduce against $\lin$ processes, and $\un$ processes
only reduce against $\un$ processes,
i.e., the reduction rules for $\mathcal{M}_0$ are those for
mixed processes, except for \rlinun and \runlin
(Figure~\ref{fig:mixed-sessions}).
The function
$\rmap{\cdot}: \mixedProcesses\longrightarrow \classicalProcesses$ in
Figure~\ref{fig-translation-mixed2classical} denotes a translation from
mixed choice processes in $\mathcal{M}_0$ to classical processes in
$\mathcal{C}$. We overload the notation and denote by $\rmap{\cdot}$
the encoding of both types (Figure~\ref{fig:embedding}) and processes
(Figure~\ref{fig-translation-mixed2classical}).


We start by addressing typing criteria.
The \emph{type preservation} criterion requires that
$\rmap{\operatorname{op}(T_1,\dots,T_n)} =
\operatorname{op}(\rmap{T_1},\dots,\rmap{T_n})$. Our encoding, in
Figure~\ref{fig:embedding}, can be called weakly type preserving in
the sense that we preserve the direction of type operations, but not the
exact type operator. For example, a $\un\oplus$ type is
translated in a $\un!$ type (and $\un\&$ type is translated in
$\un?$). Both $\oplus$ and $!$ can be seen as output types (and $\&$
and $?$ as input), so that direction is preserved.

We now move to \emph{type soundness}, but before we need to be able to
type the $\NDchoice$ operator.

\begin{lemma}
  \label{lem:ndchoice}
  The following is an admissible typing rule for typing $\NDchoice$.
  \begin{equation*}
    \frac{
      \isProc{P_i}
      \quad
      i\in I
    }{
      \isProc{\NDchoice\{P_i\}_{i\in I}}
    }
  \end{equation*}
\end{lemma}
\begin{proof}
  The typing derivation of the expansion of $\NDchoice$ leaves open the
  derivations for $\isProc{P_i}$.
\end{proof}

The type soundness theorem for our translation is item~\ref{item:p}
below; the remaining items help in building the main result.

\begin{theorem}[Type Soundness]\
 \begin{enumerate}
 \item\label{item:t} If $\UN T$, then $\UN\rmap T$.
 \item\label{item:g} If $\UN \Gamma$, then $\UN\rmap\Gamma$.
 \item\label{item:s} If $\isSubt ST$, then $\isSubt[\rmap\Theta]{\rmap S}{\rmap T}$
 \item\label{item:v} If $\isValue vT$, then $\isValue[\rmap\Gamma] v {\rmap T}$.
 \item\label{item:p} If $\isProc P$, then $\isProc[\rmap\Gamma]{\rmap{\isProc P}}$.
 \end{enumerate}
\end{theorem}
\begin{proof}
  \ref{item:t}: By case analysis on $T$ and the fact that types are
  contractive.
  \ref{item:g}: By induction on $\Gamma$ using case \ref{item:t}.
  \ref{item:s}: By coinduction on the hypothesis.
  \ref{item:v}: By rule induction on the hypothesis using items \ref{item:g}
  and \ref{item:s}.
  \ref{item:p}: By coinduction on the hypothesis, using items \ref{item:g}
    and \ref{item:v}, and lemma~\ref{lem:ndchoice}.
\end{proof}

The \emph{syntax preservation} criterion consists of ensuring that
parallel composition is translated into parallel composition and that
name restriction is translated into name restriction, which is
certainly the case with our translation. It further requires the
translation to be name invariant.  Our encoding transforms each
channel end in itself and hence is trivially name invariant.
We conclude that our translation is syntax preserving.


We now address the criteria related to the operational semantics.
We denote by $\msred$ the reflexive and transitive closure of the
reduction relations, $\osred$, in both the source and target
languages.
Sometimes we use subscript $\mixedProcesses$ to denote the reduction
of mixed choice processes and the subscript $\classicalProcesses$ for
the reduction of classical processes, even though it should be clear
from context.
The behavioral
equivalence $\asymp$ for classical sessions we are interested in
extends structural congruence $\equiv$ with the following rule
\begin{equation*}
  \NR{ab}\prod_{i\in I} \SELECT{a}{l_i}{\INACT} \;\asymp\; \INACT.
\end{equation*}
The new rule
allows collecting processes that are left by the encoding of
non-deterministic choice. We call it \emph{extended structural
  congruence}.
The following lemma characterizes the reductions of $\NDchoice$
processes: they reduce to one of the processes that are to be chosen
and leave an inert term $G$.
\begin{lemma}
  \label{lem:ndchoice_red}
  $\NDchoice\{P_i\}_{i\in I} \osred P_k \PAR G \asymp P_k$, for any $k\in I$.
\end{lemma}
\begin{proof}
  $\NDchoice\{P_i\}_{i\in I} \osred P_k \PAR G$, where
  $G=\NR{st} \prod_{i\in I}^{i\neq k} \SELECT{s}{l_i}{\INACT}$ and
  $G \asymp \INACT$.
\end{proof}


We now turn our attention to \emph{barbs} and \emph{barb
  preservation}. We say that a typed classical session process $P$
\emph{has a barb in} $x$, notation $\isProc{\barb{x}}$, if
$\isProc{P}$ and
\begin{itemize}
\item either $P\equiv \NR{x_ny_n}\ldots\NR{x_1y_1} (x\OUT{v} Q \PAR R)$
  where $x\not \in \{x_i, y_i\}_{i=1}^n$
\item or $P\equiv \NR{x_ny_n}\ldots\NR{x_1y_1} (\SELECT{x}{l} Q \PAR R)$
  where $x\not \in \{x_i, y_i\}_{i=1}^n$.
\end{itemize}
On the other hand, we say that a typed mixed session process $P$
\emph{has a barb in} $x$, notation $\isProc{\barb{x}}$, if
$\isProc{P}$ and
$P\equiv \NR{x_ny_n}\ldots\NR{x_1y_1} (\CHOOSE qx{i\in I}{M_i} \PAR
R)$ where $x\not \in \{x_i, y_i\}_{i=1}^n$ and
$\isValue{x}{\CHOICE q\oplus {i\in I} {U_i}}$. Notice that only types
can discover barbs in processes since internal choice is
indistinguishable from external choice at the process level in
$\mathcal{M}_0$.

The processes with \emph{weak barbs} are those which reduce to a
barbed process: we say that a process $P$ \emph{has a weak barb in}
$x$, notation $\isProc{\wbarb{x}}$, if $P\transred P'$ and
$\isProc[\Gamma']{\barb[P']{x}}$.

The following theorem fulfills the \emph{barb preservation criterion}: if a
mixed process has a barb, its translation has a weak barb on the same
channel.
\begin{theorem}[Barb Preservation] The translation $\translation$
  preserves barbs, that is, if $\isProc{\barb{x}}$, then
  $\isProc[\rmap{\Gamma}]{} \wbarb[\rmap{\isProc[\Gamma]{P}}]{x}$.
\end{theorem}
\begin{proof}
  An analysis of the translations of processes with barbs. In the case
  that $x$ is linear, rearranging the choice in $P$ in fragments, we
  obtain that
  $P\equiv \NR{x_ny_n}\ldots\NR{x_1y_1} (\LIN x\sum_{i\in I} (\fragr +
  \frags) \PAR R)$ and so its translation is
\begin{align*}\rmap{\isProc P}\equiv \NR{x_ny_n}\ldots\NR{x_1y_1}
  (\mathsf{NDChoice}\{ \enspace &
   \SELECT{x}{l_i^!}\NDchoice{\SEND{x}{v_{ij}}{\rmap{\isProc[\Gamma_3,
         x\colon T_i]{P_{ij}}}}}_{j\in J},
  \\
  & \SELECT{x}{l_i^?}\NDchoice{\RECEIVE{x}{y_{ik}}{\rmap{\isProc[\Gamma_2\csplit\Gamma_3,x\colon
          T'_i, y_{ik}\colon S'_i]{P'_{ik}}}}}_{k\in K}
          \}_{i\in I} \PAR \\
          &\rmap{\isProc[\Gamma'] R}).
\end{align*}

This process makes internal reduction steps in the resolution of the
outermost $\NDchoice$, non-deterministically choosing one of the
possible fragments, via Lemma~\ref{lem:ndchoice_red}.  However,
independently of which branch is chosen, they are all of the form
$\SELECT{x}{\ell}C$, which has a barb in $x$. That is:
$\rmap{\isProc P} \transred \NR{x_ny_n}\ldots\NR{x_1y_1}
(\SELECT{x}{\ell}C \PAR \rmap{\isProc[\Gamma'] R} \PAR G)$, which has a
barb in $x$. The $G$ term is the inert remainder of the $\NDchoice$
reduction.
In the unrestricted case, we have
$P\equiv \NR{x_ny_n}\ldots\NR{x_1y_1} (\UN x\sum_{i\in I} (\fragr +
\frags) \PAR R)$. The translation is
\begin{align*}\rmap{\isProc P}\equiv  \NR{x_ny_n}\ldots & \NR{x_1y_1}
  (\NR{uv}
  ( \SENDn{u}{()} \PAR \UN\RECEIVE{v}{\_} \mathsf{NDChoice}\{ & \\
& \NR{ab} \SEND{x}{a} \SELECT{b}{l_i^!}\NDchoice{\{\SEND{b}{v_{ij}}
(\SENDn{u}{()} \PAR {\rmap{\isProc[\Gamma_1
       ]{P_{ij}}})}\}}_{j\in J}, & \\
& \NR{ab} \SEND{x}{a} \SELECT{b}{l_i^?}
    \NDchoice{\{\RECEIVE{b}{y_{ik}}
    (\SENDn{u}{()} \PAR
    {\rmap{\isProc[\Gamma_1, y_{ik}\colon S_i']{P'_{ik}}})}\}}_{k\in K}
    \}_{i\in I} ) \PAR \rmap{\isProc[\Gamma_2] R}) .
\end{align*}
The process  starts by reducing via \runcom on the $u,v$ channels to the process
\begin{align*}\rmap{\isProc P}\transred  \NR{x_ny_n}\ldots & \NR{x_1y_1}
  \NR{uv}
  ( \mathsf{NDChoice}\{ & \\
& \NR{ab} \SEND{x}{a} \SELECT{b}{l_i^!}\NDchoice{\{\SEND{b}{v_{ij}}
(\SENDn{u}{()} \PAR {\rmap{\isProc[\Gamma_1
       ]{P_{ij}}})}\}}_{j\in J}, & \\
& \NR{ab} \SEND{x}{a} \SELECT{b}{l_i^?}
    \NDchoice{\{\RECEIVE{b}{y_{ik}}
    (\SENDn{u}{()} \PAR
    {\rmap{\isProc[\Gamma_1, y_{ik}\colon S_i']{P'_{ik}}})}\}}_{k\in K}
    \}_{i\in I} \PAR \rmap{\isProc[\Gamma_2] R} \PAR U)
\end{align*}
where $U$ is the persistent part of the unrestricted process.
This process, in turn, reduces via the $\NDchoice$ (Lemma~\ref{lem:ndchoice_red}) to one
of the possible branches which are all of the form $\NR{ab}x\OUT{a}C$,
\begin{align*}\rmap{\isProc P}\transred  \NR{x_ny_n}\ldots & \NR{x_1y_1}
  \NR{uv}
  \NR{ab} (\SEND{x}{a} C ) \PAR \rmap{\isProc[\Gamma'] R} \PAR U  \PAR G).
\end{align*}
Since $P$ has a barb in $x$, $x\not \in \{x_i, y_i\}_{i=1}^n$ and so this process also
has a barb in $x$, concluding that $\rmap{\isProc P}$ has indeed a weak barb in $x$.
\end{proof}


Finally, we look at operational completeness.
Operational completeness relates the behavior of mixed sessions
against their classical sessions images: any reduction step in mixed
sessions can be mimicked by a sequence of reductions steps in
classical sessions, modulo extended structural congruence. The ghost
reductions result from the new channels and communication inserted by
the translation, namely those due to the $\NDchoice$ and to the
encoding of ``loops'' for $\un$ mixed choices.

\begin{theorem}[Reduction Completeness]
\label{thm:reduction_completeness}
The translation $\translation$
  is \emph{operationally complete}, that is, if
  $P \osred_\mixedProcesses P'$, then
  $\rmap {\isProc P}
  \msred_\classicalProcesses\asymp_\classicalProcesses
  \rmap{\isProc{P'}}$,
\end{theorem}
\begin{proof}
  By rule induction on the derivation of
  $P \osred_\mixedProcesses P'$. We detail two cases.

  Case \rpar. We can show that
  %
    if $Q_1 \msred_\classicalProcesses Q_1'$, then
    $Q_1\PAR Q_2 \msred_\classicalProcesses Q_1' \PAR Q_2$,
  by induction on the length of the reduction.  Then we have
  $\rmap {\isProc[\Gamma] {P_1 \PAR P_2} } = \rmap {\isProc[\Gamma_1]
    { P_1 } } \PAR \rmap {\isProc[\Gamma_2] { P_2 } }$ with
  $\Gamma = \Gamma_1\csplit \Gamma_2$.  By induction we have
  $\rmap {\isProc[\Gamma_1]{P_1}} \msred_\classicalProcesses
  Q\asymp_\classicalProcesses \rmap{\isProc[\Gamma_1]{P_1'}}$.  Using
  the above result and the fact that $\asymp_\classicalProcesses$ is a
  congruence, we get
  $\rmap {\isProc[\Gamma_1] { P_1 } } \PAR \rmap {\isProc[\Gamma_2] {
      P_2 } } \msred_\classicalProcesses Q \PAR \rmap
  {\isProc[\Gamma_2] { P_2 }} \asymp_\classicalProcesses
  \rmap{\isProc[\Gamma_1]{P_1'}} \PAR \rmap {\isProc[\Gamma_2] {P_2}}
  = \rmap {\isProc{P'_1 \PAR P_2}}$.
  The cases for \rres and \rstruct are similar.

  \begin{sloppypar}
    Case \rlinlin. Let
    $\Gamma, x\colon R,y\colon S = \Gamma' \csplit \Gamma'' \csplit
    \Gamma'''$ and
    $\isValue[\Gamma']{x}{\CHOICE{\lin}{\&}{}{\BRANCH{l}{!}{T_0}{R_0},\ldots}}$
    and
    $\isValue[\Gamma'']{y}{\CHOICE{\lin}{\oplus}{}{\BRANCH{l}{?}{U_0}{S_0},
        \ldots}}$, with $\areEquiv{T_0}{U_0}$ and
    $\areDual{R_0}{S_0}$.
    Let $\Gamma' = \Gamma_1' \csplit \Gamma_2' \csplit \Gamma_3'$ and
    $\Gamma'' =\Gamma_1'' \csplit \Gamma_2'' \csplit \Gamma_3''$.
    We have:
  \end{sloppypar}
  \begin{align*}
  	&\rmap{\isProc[\Gamma]{\NR{xy}(\lin x(\BRANCHP {l}!v{P} + M) \PAR \lin y(\BRANCHP {l}?z{Q} + N) \PAR O)}}
\\
  	=&  \NR{xy}(
  		\TBRANCH{x}{l^?}{\NDchoice\{ \SEND{x}{v}\rmap{\isProc[\Gamma_3',x:R_0]P}, \ldots\}, \ldots}{} \PAR\\
  		&\qquad \enspace\NDchoice\{\SELECT{y}{l^?}\NDchoice\{\RECEIVE{y}{z}\rmap{\isProc[(\Gamma_2''\csplit \Gamma_3'', y:S_0, z:U_0)]Q},
  	 	\ldots\}, \ldots\}  	 \PAR
      \rmap{\isProc[\Gamma''']{O}})\\
      \osred \asymp & \NR{xy}(
  		\TBRANCH{x}{l^?}{\NDchoice\{ \SEND{x}{v}\rmap{\isProc[\Gamma_3',x:R_0]P},  \ldots\}, \ldots}{} \PAR\\
  		&\qquad \enspace\SELECT{y}{l^?}\NDchoice\{\RECEIVE{y}{z}\rmap{\isProc[(\Gamma_2''\csplit \Gamma_3'', y:S_0, z:U_0)]Q},
      \ldots\}  	  \PAR
  		\rmap{\isProc[\Gamma''']{O}})\\
      \osred & \NR{xy}(
  		{\NDchoice\{ \SEND{x}{v}\rmap{\isProc[\Gamma_3',x:R_0]P}, \ldots\}} \PAR\\
  		&\qquad \enspace\NDchoice\{\RECEIVE{y}{z}\rmap{\isProc[(\Gamma_2''\csplit \Gamma_3'', y:S_0, z:U_0)]Q},
      \ldots\}  	  \PAR
  		\rmap{\isProc[\Gamma''']{O}})
\\
  	\osred \osred\asymp & \NR{xy}(
  		{\SEND{x}{v}\rmap{\isProc[\Gamma_3',x:R_0]P}}  \PAR
  		\RECEIVE{y}{z}\rmap{\isProc[(\Gamma_2''\csplit \Gamma_3'', y:S_0, z:U_0)]Q}
  	 	  	  \PAR
  		\rmap{\isProc[\Gamma''']{O}})
    \end{align*}
    \begin{align*}
  	 \osred & \NR{xy}(
  		{\rmap{\isProc[\Gamma_3',x:R_0]P}}  \PAR
  		\rmap{\isProc[\Gamma_2''\csplit \Gamma_3'', y:S_0, z:U_0]Q }\subs vz
  	 	  	  \PAR
  		\rmap{\isProc[\Gamma''']{O}})
\\
  	  = & \NR{xy}(
  		{\rmap{\isProc[\Gamma_3',x:R_0]P}}  \PAR
  		\rmap{\isProc[\Gamma_2'\csplit \Gamma_2''\csplit \Gamma_3'', y:S_0]Q\subs vz }
  	 	  	 \PAR
  		\rmap{\isProc[\Gamma''']{O}})
\\
   	= &	\rmap{  \isProc[\Gamma_2'\csplit \Gamma_3'\csplit\Gamma_2''\csplit \Gamma_3''\csplit \Gamma''']
  		{\NR{xy} (P  \PAR Q\subs vz \PAR O)} }
\\
   	= &	\rmap{  \isProc {\NR{xy} (P  \PAR Q\subs vz \PAR O)} }
  \end{align*}
  %
  Notice that $\Gamma_1' = \Delta_1,x\colon R$
  where $\Delta_1$ is $\un$, hence $\Delta_1$ is in $\Gamma_2'$ and in
  $\Gamma'_3$. The same reasoning applies to $\Gamma_1'' $. Since
  context $\Gamma_2'$ is used to type $v$, the substitution
  lemma~\cite{DBLP:journals/iandc/Vasconcelos12} reintroduces it in
  the context for $Q\subs vz$.

  The case for \runun is similar, albeit more verbose.  The cases for
  \rift and \riff are direct.
\end{proof}

We can show that the translation does \emph{not} enjoy reduction
soundness. Consider the classical process $Q$ to be the encoding of
process $P$ of the form $\un y(\BRANCHP m ? z \INACT)$, described in
the right part of Figure~\ref{fig:example3_translation}. Soundness
requires that if $Q \osred_\classicalProcesses Q'$, then
$P \msred_\mixedProcesses P'$ and
$Q \msred_\classicalProcesses \asymp_\classicalProcesses
\rmap{\isProc{P'}}$. Clearly, $Q$ has an initial reduction step (on
channel $u_2v_2$), which cannot be mimicked by $P$.  But this
reduction is a transition internal to process~$Q$, a $\tau$
transition. Equipped with a suitable notion of labelled transition
systems on both languages that include $\tau$ transitions, and by
using a weak bisimulation that ignores such transitions, we expect
soundness to hold.


\section{Further Work}
\label{sec:conclusion}

There are two avenues that may be followed. One extends the encoding
to the full language of mixed sessions, by taking into consideration
the axioms in the reduction relation that match $\lin$ choices against
$\un$ choices. The other pursues semantic
preservation~\cite{DBLP:journals/iandc/KouzapasPY19} by establishing a
full abstraction result, requiring the development of typed
equivalences for the two languages.


\paragraph{Acknowledgements}
This work was supported by FCT through the LASIGE Research Unit,  ref.\
UIDB/ 00408/2020, and by Cost Action CA15123 EUTypes.

\bibliographystyle{eptcs}
\bibliography{biblio}

\appendix

\section{The Syntax, Operational Semantics, and Type System of Mixed
  and Classical Sessions}

\paragraph{Mixed Sessions}

\begin{figure}[!t]
    \emph{Mixed syntactic forms}
    \begin{align*}
        v \grmeq & & \text{Values:}\\
        & x  & \text{variable}\\
        & \truek \grmor \falsek & \text{boolean values}\\
        & () & \text{unit value}\\
        P  \grmeq  & & \text{Processes:}\\
    & \CHOOSE qx{i\in I}{M_i} & \text{choice}\\
        & P\PAR P & \text{parallel composition}\\
        & \NR{xx}P & \text{scope restriction}\\
        & \IF vPP & \text{conditional}\\
        & \INACT & \text{inaction}\\
       %
       %
       M \grmeq & & \text{Branches:}\\
       & \BRANCHP l \star vP & \text{branch}\\
        \star \grmeq & & \text{Polarities:}\\
        & ! \grmor ?  & \text{out and in}\\
        q  \grmeq & & \text{Qualifiers:}\\
        & \lin & \text{linear}\\
        & \un & \text{unrestricted}
     \end{align*}
     \emph{Structural congruence}, $P \equiv P$
  \begin{gather*}
   P\PAR Q \equiv Q\PAR P
    \qquad
    (P\PAR Q)\PAR R \equiv P\PAR (Q\PAR R)
    \qquad
    P\PAR \INACT\equiv P
   \\
    \NR{xy} P\PAR Q \equiv \NR{xy}(P\PAR Q)
   \qquad
   \NR{xy}\INACT \equiv \INACT
   \qquad
   \NR{{wx}}\NR{yz}P\equiv\NR{{yz}}\NR{wx}P
 \end{gather*}
   \emph{Mixed reduction rules}, $P \osred P$
  \begin{gather*}
   \tag*{\rift\riff}
   \IF{\truek}{P}{Q} \osred P
   \qquad
    \IF{\falsek}{P}{Q} \osred Q
   \\
   \tag*\rlinlin
    \NR{xy}(\lin x(\BRANCHP l!vP + M) \PAR \lin y(\BRANCHP l?zQ + N) \PAR R)
    \osred
    \NR{xy}(P \PAR Q\subs v z \PAR R)
   \\
   \tag*\rlinun
    \NR{xy}(\lin x(\BRANCHP l!vP + M) \PAR \un y(\BRANCHP l?zQ + N) \PAR R)
    \osred
    \NR{xy}(P \PAR Q\subs v z \PAR \un y(\BRANCHP l?zQ + N) \PAR R)
   \\
   \tag*\runlin
    \NR{xy}(\un x(\BRANCHP l!vP + M) \PAR \lin y(\BRANCHP l?zQ + N) \PAR R)
    \osred
    \NR{xy}(P \PAR Q\subs v z \PAR \un x(\BRANCHP l!vP + M) \PAR R)
   \\
   \tag*\runun
    \NR{xy}(\un x(\BRANCHP l!vP + M) \PAR \un y(\BRANCHP l?zQ + N) \PAR R)
    \osred\\
    \qquad\qquad\qquad
    \NR{xy}(P \PAR Q\subs v z \PAR \un x(\BRANCHP l!vP + M) \PAR
    \un y(\BRANCHP l?zQ + N) \PAR R)
   \\
   \tag*{\rres\rpar \rstruct}
    \frac{
      P \osred  Q
    }{
      \NR{xy}P \osred \NR{xy}Q
    }
    \qquad
    \frac{
      P \osred  Q
    }{
      P \PAR R \osred Q \PAR R
    }
    \qquad
    \frac{
      P \equiv P' \qquad P' \osred Q' \qquad Q' \equiv Q
    }{
      P \osred Q
    }
  \end{gather*}

     \caption{Mixed session types: process syntax and reduction}
     \label{fig:mixed-sessions}
   \end{figure}


\begin{figure}[!t]
  \begin{align*}
    T \grmeq & & \text{Types:}\\
             & \CHOICE q\sharp {i\in I} {U_i} & \text{choice}\\
             & \End & \text{termination}\\
             & \unit \grmor \bool & \text{unit and boolean}\\
             & \rect & \text{recursive type}\\
             & a & \text{type variable}\\
    U \grmeq & & \text{Branches:}\\
             & \BRANCH l\star TT & \text{branch}\\
    \sharp \grmeq & & \text{Views:}\\
             & \oplus \grmor \& & \text{internal and external}\\
    \Gamma \grmeq & & \text{Contexts:}\\
             & \Empty & \text{empty}\\
             & \Gamma, x\colon T & \text{entry}
  \end{align*}
  \emph{The \un predicate}, $\UN T$, $\UN\Gamma$
  \begin{equation*}
    \UN(\CHOICE \un\sharp {i\in I} {U_i})
    \quad
    \UN(\rect) \text{ if } \UN T
    \quad
    \UN(\End,\unit,\bool)
    \qquad\qquad
    \UN\Empty
    \quad
    \UN(\Gamma,x\colon T) \text{ if } \UN\Gamma \wedge \UN T
  \end{equation*}
  \emph{Branch subtyping}, $\isSubt UU$
  \begin{gather*}
    \frac{
      \isSubt{S_2}{S_1}
      \quad
      \isSubt{T_1}{T_2}
    }{
      \isSubt{\BRANCH l!{S_1}{T_1}}{\BRANCH l!{S_2}{T_2}}
    }
    \qquad
    \frac{
      \isSubt{S_1}{S_2}
      \quad
      \isSubt{T_1}{T_2}
    }{
      \isSubt{\BRANCH l?{S_1}{T_1}}{\BRANCH l?{S_2}{T_2}}
    }
  \end{gather*}
  \emph{Coinductive subtyping rules}, $\isSubt TT$
  \begin{gather*}
    \frac{}{\isSubt \End \End}
    \qquad
    \frac{}{\isSubt \unit \unit}
    \qquad
    \frac{}{\isSubt \bool \bool}
    \qquad
    \frac{
      \isSubt[\Theta, (\REC aS, T)] {S\subs{\REC aS}{a}} T
    }{
      \isSubt{\REC aS}T
    }
    \qquad
    \frac{
      \isSubt[\Theta, (S, \REC aT)] S {T\subs{\REC aT}{a}}
    }{
      \isSubt S {\REC aT}
    }
    \\
    \frac{
      J\subseteq I
      \qquad
      \isSubt {U_j}{V_j}
    }{
      \isSubt {\CHOICE q\oplus{i\in I}{U_i}}{\CHOICE q\oplus{j\in J}{V_j}}
    }
    \qquad
    \frac{
      I\subseteq J
      \qquad
      \isSubt {U_i}{V_i}
    }{
      \isSubt {\CHOICE q\&{i\in I}{U_i}}{\CHOICE q\&{j\in J}{V_j}}
    }
  \end{gather*}
  \emph{Polarity duality and view duality}, $\areDualP\sharp\sharp$
  and $\areDualP\star\star$
  \begin{equation*}
    \areDualP !?
    \qquad
    \areDualP ?!
    \qquad
    \qquad
    \areDualP \oplus \&
    \qquad
    \areDualP \& \oplus
  \end{equation*}
  \emph{Coinductive type duality rules}, $\areDual TT$
  \begin{gather*}
    \frac{}{\areDual\End\End}
    \qquad
    \frac{
      \sharp\, \bot\, \flat
      \qquad
      \star_i \bot \bullet_i
      \qquad
      \areEquiv[\Empty] {S_i}{S'_i}
      \qquad
      \areDual {T_i}{T'_i}
    }{
      \areDual {\CHOICE q{\sharp}{i\in I}{\BRANCH{l_i}{\star}{S_i}{T_i}}}{\CHOICE q{\flat}{i\in I}{\BRANCH{l_i}{\bullet}{S'_i}{T'_i}}}
    }
    \\
    \frac{
    \areDual[\Theta, (\REC aS, T)] {S\subs{\REC aS}{a}} T
    }{
    \areDual{\REC aS}T
    }
    \qquad
    \frac{
    \areDual[\Theta, (S, \REC aT)] S {T\subs{\REC aT}{a}}
    }{
    \areDual S {\REC aT}
    }
  \end{gather*}
  \caption{Mixed session types: types syntax, subtyping, and duality}
  \label{fig:mixed-sessions2}
\end{figure}


\begin{figure}[!t]
    \emph{\un and \lin predicates}, $\un(T)$, $\lin(T)$
  \begin{equation*}
    \un(\End)
    \quad
    \un(\unit)
    \quad
    \un(\bool)
    \quad
    \un(\CHOICE \un\sharp{}{U_i})
    \quad
    \frac{\un(T)}{\un(\mu a.T)}
    \qquad\qquad
    \frac{}{\lin(T)}
  \end{equation*}
  \emph{Context split}, $\Gamma = \Gamma \csplit \Gamma$
  \begin{gather*}
    \Empty = \Empty \csplit \Empty
    \qquad \qquad
    \frac{
      \Gamma_1 \csplit \Gamma_2 = \Gamma
      \qquad
      \un(T)
   }{
      \Gamma, x\colon T = (\Gamma_1,x\colon T) \csplit
      (\Gamma_2,x\colon T)
    }
    \\
    \frac{
      \Gamma = \Gamma_1 \csplit \Gamma_2
    }{
      \Gamma, x\colon \lin\,p = (\Gamma_1,x\colon \lin\,p) \csplit \Gamma_2
    }
    \qquad \qquad
    \frac{
      \Gamma = \Gamma_1 \csplit \Gamma_2
  }{
      \Gamma, x\colon \lin\,p = \Gamma_1 \csplit (\Gamma_2,x\colon \lin\,p)
    }
  \end{gather*}
  \emph{Context update}, $\Gamma + x\colon T = \Gamma$
  \begin{equation*}
    \frac{
      x\colon U \notin \Gamma
    }{
      \Gamma + x \colon T = \Gamma, x \colon T
    }
    \qquad
    \frac{
      \un(T) \qquad \areEquiv{T}{U}
    }{
      (\Gamma, x\colon T) + x \colon U = (\Gamma, x\colon T)
    }
  \end{equation*}
  \emph{Typing rules for values}, $\isValue vT$
  \begin{gather*}
    \tag*{\tunit\ttrue\tfalse\tvar\tsubt}
    \frac{
      \un(\Gamma)
    }{
      \isValue {()} \unit
    }
    \quad\;\;
    \frac{
      \un(\Gamma)
    }{
      \isValue {\truek, \falsek} \bool
    }
    \quad\;\;
    \frac{
      \un(\Gamma_1,\Gamma_2)
    }{
      \isValue[\Gamma_1,x\colon T,\Gamma_2] x T
    }
    \qquad
    \frac{
      \isValue vS
      \quad
      \isSubt[\Empty] ST
    }{
      \isValue vT
    }
  \end{gather*}
  \emph{Typing rules for branches}, $\isBranch MU$
  \begin{gather*}
    \tag*{\tout\tin}
    \frac{
      \isValue[\Gamma_1] vS
      \qquad
      \isProc[\Gamma_2] P
    }{
      \isBranch[\Gamma_1 \csplit \Gamma_2] {\BRANCHP l!vP} {\BRANCH l!ST}
    }
    \qquad
    \frac{
      \isProc[\Gamma,x\colon S] P
    }{
      \isBranch {\BRANCHP l?xP} {\BRANCH l?ST}
    }
  \end{gather*}
  \emph{Typing rules for processes}, $\isProc P$
  \begin{gather*}
    \tag*{\tchoice}
    \frac{
      q_1(\Gamma_1 \!\csplit\! \Gamma_2)
      \;\;\;
      \isValue[\Gamma_1]{x}{\CHOICE{q_2}{\sharp}{i\in I}{\BRANCH{l_i}{\star}{S_i}{T_i}}}
      \;\;\;
      \isBranch[\Gamma_2 + x\colon T_j] {\BRANCHP{l_j}{\star}{v_j}{P_j}}{\BRANCH{l_j}{\star}{S_j}{T_j}}
      \;\;\;
      \{l_j^\star\}_{j \in J} = \{l_i^\star\}_{i \in I}
    }{
      \isProc[\Gamma_1 \csplit \Gamma_2]{\CHOOSE {q_1}x{j\in J}{\BRANCHP{l_j}{\star}{v_j}{P_j}}}
    }
    \\
    \tag*{\tinact\tpar\tif\tres}
    \frac{
      \un(\Gamma)
    }{
      \isProc \INACT
    }
    \qquad
    \frac{
      \isProc[\Gamma_1] P
      \quad
      \isProc[\Gamma_2] Q
    }{
      \isProc[\Gamma_1 \csplit \Gamma_2] {P \PAR Q}
    }
    \qquad
    \frac{
      \Gamma_1 \vdash v\colon \bool
      \quad
      \Gamma_2 \vdash P
      \quad
      \Gamma_2 \vdash Q
    }{
      \Gamma_1 \csplit \Gamma_2 \vdash \ifp
    }
    \qquad
    \frac{
      \areDual[\Empty] ST
      \quad
      \isProc[\Gamma,x\colon S,y\colon T] P
    }{
      \isProc{\NR{xy}P}
    }
  \end{gather*}
  \caption{Mixed session types: \un and \lin predicates, context split
  and update, and typing}
  \label{fig:mixed-sessions3}
\end{figure}


The syntax of process and the operational semantics are in
Figure~\ref{fig:mixed-sessions}.
The syntax of types, and the notions of subtyping and type duality are
in Figure~\ref{fig:mixed-sessions2}.
The \un and \lin predicates, the context split and update operations,
and the typing rules are in Figure~\ref{fig:mixed-sessions3}.

\paragraph{Classical Sessions}

\begin{figure}[!t]
 \emph{Syntactic forms}
 \begin{align*}
    P  \grmeq  & \dots & \text{Processes:}\\
    & \SEND xvP & \text{output}\\
    & q\RECEIVE xxP & \text{input}\\
    & \selp & \text{selection}\\
    & \branchp & \text{branching}\\
    T \grmeq & \dots & \text{Types:}\\
    & q\star T.T & \text{communication}\\
    & q\sharp\{l_i\colon T_i\}_{i\in I} & \text{choice}
  \end{align*}
  \emph{Reduction rules}, $P \osred P$, (plus \rres\rpar
  \rstruct from Figure~\ref{fig:mixed-sessions})
  \begin{gather*}
    \tag*\rlincom
    \NR{xy}(\SEND x v P \PAR \lin\,\RECEIVE y z Q \PAR R)
    \osred
    \NR{xy}(P \PAR Q\subs v z \PAR R)
    \\
    \tag*\runcom
    \NR{xy}(\SEND x v P \PAR \un\,\RECEIVE y z Q \PAR R)
    \osred
    \NR{xy}(P \PAR Q\subs v z \PAR \un\,\RECEIVE yzQ \PAR R)
    \\
    \tag*\rcase
    \frac{
      j\in I
    }{
      \NR{xy}(\SELECT{x}{l_j}{P} \PAR \TBRANCH{y}{l_i}{Q_i}{i\in I} \PAR
      R)
      \osred
      \NR{xy}(P \PAR Q_j \PAR R)
    }
\end{gather*}
  \emph{Subtyping rules}, $\isSubt TT$
  \begin{gather*}
    \frac{
      \isSubt TS
      \qquad
      \isSubt{S'}{T'}
    }{
      \isSubt{q!S.S'}{q!T.T'}
    }
    \qquad
    \frac{
      \isSubt ST
      \qquad
      \isSubt{S'}{T'}
    }{
      \isSubt{q?S.S'}{q?T.T'}
    }
    \\
    \frac{
      J\subseteq I
      \qquad
      \isSubt{S_j}{T_j}
    }{
      \isSubt{q\oplus\{l_i\colon S_i\}_{i\in I}}{q\oplus\{l_j\colon T_j\}_{j\in J}}
    }
    \qquad
    \frac{
      I\subseteq J
      \qquad
      \isSubt{S_i}{T_i}
    }{
      \isSubt{q\&\{l_i\colon S_i\}_{i\in I}}{q\&\{l_j\colon T_j\}_{j\in J}}
    }
  \end{gather*}
  %
  \emph{Type duality rules}, $\areDual TT$
  \begin{gather*}
    \frac{
      \isSubt[\Empty]ST
      \qquad
      \isSubt[\Empty]TS
      \qquad
      \areDual{S'}{T'}
    }{
      \areDual{q?S.S'}{q!T.T'}
    }
    \qquad
    \frac{
      \areDual{S_i}{T_i}
    }{
      \areDual{q\oplus\{l_i\colon S_i\}_{i\in I}}{q\&\{l_i\colon T_i\}_{i\in I}}
    }
  \end{gather*}
  \emph{Typing rules}, $\isProc P$, (plus \tinact\tpar\tres from Figure~\ref{fig:mixed-sessions2})
  \begin{gather*}
    \tag*{\ttout}
    \frac{
      \Gamma_1 \vdash x\colon q\,\OUT TU
      \qquad
      \Gamma_2 \vdash  v \colon T
      \qquad
      \Gamma_3 \cupdate x\colon U \vdash P
    }{
      \Gamma_1\csplit\Gamma_2\csplit\Gamma_3 \vdash \sendp
    }
    \\
    \tag*{\ttin}
    \frac{
      q_1(\Gamma_1 \csplit \Gamma_2)
      \qquad
      \Gamma_1 \vdash x\colon q_2\IN TU
      \qquad
      (\Gamma_2 \cupdate x\colon U), y\colon T \vdash P
    }{
      \Gamma_1 \csplit \Gamma_2 \vdash q_1\RECEIVE xyP
    }
    \\
        \tag*{\tbranch}
    \frac{
      \Gamma_1 \vdash x\colon q\brancht
      \qquad
      \Gamma_2\cupdate x\colon T_i \vdash P_i
      \qquad \forall i\in I
    }{
      \Gamma_1\csplit\Gamma_2 \vdash \branchp
    }
    \\
    \tag*{\tselect}
    \frac{
      \Gamma_1 \vdash x\colon q\selectt
      \qquad
      \Gamma_2 \cupdate x\colon T_j \vdash P
      \qquad j \in I
    }{
      \Gamma_1\csplit\Gamma_2 \vdash \SELECT{x}{l_j}{P}
    }
  \end{gather*}
  \caption{Classical session types}
  \label{fig:classical-sessions}
\end{figure}


The syntax, operational semantics, and type system are in
Figure~\ref{fig:classical-sessions}.



\end{document}
